\RequirePackage{amsthm}

\documentclass[sn-mathphys-num]{sn-jnl}% Math and Physical Sciences Numbered Reference Style 
\usepackage{graphicx}%
\usepackage{multirow}%
\usepackage{amsmath,amssymb,amsfonts}%
\usepackage{amsthm}%
\usepackage{mathrsfs}%
\usepackage[title]{appendix}%
\usepackage{xcolor}%
\usepackage{textcomp}%
\usepackage{manyfoot}%
\usepackage{booktabs}%
\usepackage{algorithm}%
\usepackage{algorithmicx}%
\usepackage{algpseudocode}%
\usepackage{listings}%
\usepackage{anyfontsize}
\usepackage{makecell}
\usepackage{mathtools}
\usepackage{rotating}
\usepackage{tabularx}
\usepackage{caption}
\usepackage{setspace}
\usepackage{pdflscape}

\DeclarePairedDelimiter\ceil{\lceil}{\rceil}
\newcolumntype{Y}{>{\centering\arraybackslash}X}

\theoremstyle{thmstyleone}%
\newtheorem{theorem}{Theorem}%  meant for continuous numbers
\newtheorem{proposition}[theorem]{Proposition}% 

\theoremstyle{thmstyletwo}%

\theoremstyle{thmstylethree}%

\begin{document}

%%\unnumbered% uncomment this for unnumbered level heads

\title[Article Title]{
A Cholesky decomposition-based asset selection heuristic for sparse tangent portfolio optimization
}


%%%%%%%%%%%%%%%%%%%%%%%%%%%%%% ANONYMOUS %%%%%%%%%%%%%%%%%%%%%%%%%%%%%%
% \author[]{\fnm{Anonymous Author(s)} \sur{}}

%% CHECK Acknoledgement!!

%%%%%%%%%%%%%%%%%%%%%%%%%%%%%% AUTHORS %%%%%%%%%%%%%%%%%%%%%%%%%%%%%%
\author[1]{\fnm{Hyunglip} \sur{Bae}}\email{qogudflq@kaist.ac.kr}
\equalcont{These authors contributed equally to this work.}

\author[1]{\fnm{Haeun} \sur{Jeon}}\email{haeun39@kaist.ac.kr}
\equalcont{These authors contributed equally to this work.}

\author[2]{\fnm{Minsu} \sur{Park}}\email{mspark0425@kaist.ac.kr}

\author*[3, 4]{\fnm{Yongjae} \sur{Lee}}\email{yongjaelee@unist.ac.kr}

\author*[1]{\fnm{Woo Chang} \sur{Kim}}\email{wkim@kaist.ac.kr}

\affil[1]{\orgdiv{Department of Industrial and Systems Engineering}, \orgname{KAIST}}

\affil[2]{\orgdiv{Graduate School of Data Science}, \orgname{KAIST}}

\affil[3]{\orgdiv{Department of Industrial Engineering}, \orgname{UNIST}}

\affil[4]{\orgdiv{Artificial Intelligence Graduate School}, \orgname{UNIST}}
%%%%%%%%%%%%%%%%%%%%%%%%%%%%%%%%%%%%%%%%%%%%%%%%%%%%%%%%%%%%%%%%%%%%%%%%%%%%%%%%%%%%%%%%%%

% \affil*[3]{\orgdiv{Department}, \orgname{Organization}, \orgaddress{\street{Street}, \city{City}, \postcode{610101}, \state{State}, \country{Country}}}


\abstract{
In practice, including large number of assets in mean-variance portfolios can lead to higher transaction costs and management fees. To address this, one common approach is to select a smaller subset of assets from the larger pool, constructing more efficient portfolios. As a solution, we propose a new asset selection heuristic which generates a pre-defined list of asset candidates using a surrogate formulation and re-optimizes the cardinality-constrained tangent portfolio with these selected assets. This method enables faster optimization and effectively constructs portfolios with fewer assets, as demonstrated by numerical analyses on historical stock returns. Finally, we discuss a quantitative metric that can provide a initial assessment of the performance of the proposed heuristic based on asset covariance.
}




\maketitle

\section{Introduction}
\label{sec:intro}

\citet{10.2307/2975974} proposed a mean-variance model for selecting a portfolio that minimizes risk for a given level of return, which has since become a cornerstone of modern portfolio theory \citep{kim2021mean}. Diversification is a key principle in risk reduction. However, an overly diversified portfolio can present practical challenges, making it difficult for managers to oversee numerous components, which leads to higher fixed transaction costs \citep{woodside2011portfolio, kremer2020sparse}.This complicates the practical application of the Markowitz model.

To tackle this issue, a cardinality constraint was introduced to the original Markowitz model, limiting the number of assets in the portfolio. However, this constraint increases the problem's complexity, rendering it NP-hard \citep{garey1979computers}. Since finding an exact optimal portfolio in such cases demands significant computational time, recent studies have focused on developing efficient algorithms.

There are two main approaches to controlling the cardinality of mean-variance portfolios. The first approach indirectly considers cardinality by adding a penalty term to the objective function or by imposing an upper bound constraint on the portfolio weight norm. For example, \citet{brodie2009sparse} controlled cardinality by introducing a penalty term to the objective function using the \(\ell_1\)-norm, while \citet{chen2013sparse} applied the \(\ell_p\)-norm to penalize portfolio weights. Additional examples of norm regularization can be found in \citet{demiguel2009generalized, kim2014portfolio, kremer2020sparse, corsaro2019adaptive}. The second approach directly incorporates cardinality constraints into the model and develops efficient algorithms to solve the resulting problem. \citet{bienstock1996computational} proposed a branch-and-cut algorithm for a relaxation of the cardinality constraint, and \citet{li2006optimal} presented an exact solution algorithm for obtaining an optimal lot solution for a cardinality-constrained mean-variance model under concave transaction costs. \citet{gao2013optimal} explored efficient exact solution algorithms by modifying the primal objective function, while \citet{kim2016sparse} and \citet{lee2020sparse} employed semi-definite relaxation to address the cardinality problem. Heuristic methods have also been explored: \citet{chang2000heuristics} examined genetic algorithms, tabu search, and simulated annealing, comparing their performance with the efficient frontier, and noted that tradeoffs between risk, return, and asset count should be explicitly considered when selecting a portfolio. Additionally, \citet{maringer2003optimization} proposed a hybrid local search algorithm combining simulated annealing and evolutionary strategies.

While much of the prior research on sparse portfolio selection has focused either on the indirect reduction of portfolio cardinality or the development of algorithms to address this complex problem, this paper introduces a novel asset selection heuristic called OSCAR. The name OSCAR, which stands for \textit{Optimize, Select with Cholesky, And Re-optimize}, encapsulates the entire process of the heuristic. The main difficulty of the cardinality-constrained portfolio problem lies in simultaneously determining the asset selection that satisfies the cardinality constraint and performing the subsequent optimization. OSCAR addresses this challenge effectively by decoupling the problem into separate steps, utilizing the properties of the Sharpe ratio to streamline the process. Initially, OSCAR optimizes the tangent portfolio without considering the cardinality constraint, yielding a fully diversified asset allocation. This portfolio is then transformed using the Cholesky decomposition of the covariance matrix, after which the components of the resulting vector are ranked by their absolute values to identify the selected assets. Finally, a sparse tangent portfolio is optimized based on the chosen assets. The overall process of OSCAR is depicted in Figure \ref{fig:framework}. This approach enables us to maximize the Sharpe ratio under cardinality constraints in significantly less time compared to the existing solver. We demonstrate the effectiveness and efficiency of OSCAR through numerical experiments using various historical stock returns.

% 우리 방법 더 자세히 설명 및 강조
% OSCAR works by first optimizing the tangent portfolio without cardinality constraints to achieve diversification, then selecting a pre-determined set of asset candidates using the Cholesky decomposition of the covariance matrix based on the geometric properties of the Sharpe ratio, and finally re-optimizing the sparse tangent portfolio with the selected assets.

The remainder of the paper is organized as follows.
Section 2 represents our proposed asset selection method OSCAR for the cardinality-constrained Sharpe ratio maximization problem.
The numerical experiments are represented in Section 3.
We further discuss characteristics in asset sets and performance of OSCAR in Section 4.
Finally, Section 5 concludes the paper.

\section{Asset selection method for sparse tangent portfolios}
\label{sec:oscar}

%%%%%%%%%%%%%%%%%%%%%%% FIG 1 %%%%%%%%%%%%%%%%%%%%%%%
\begin{figure*}[t!]
    \begin{center}
    
    \centerline{\includegraphics[width=1.1\textwidth]{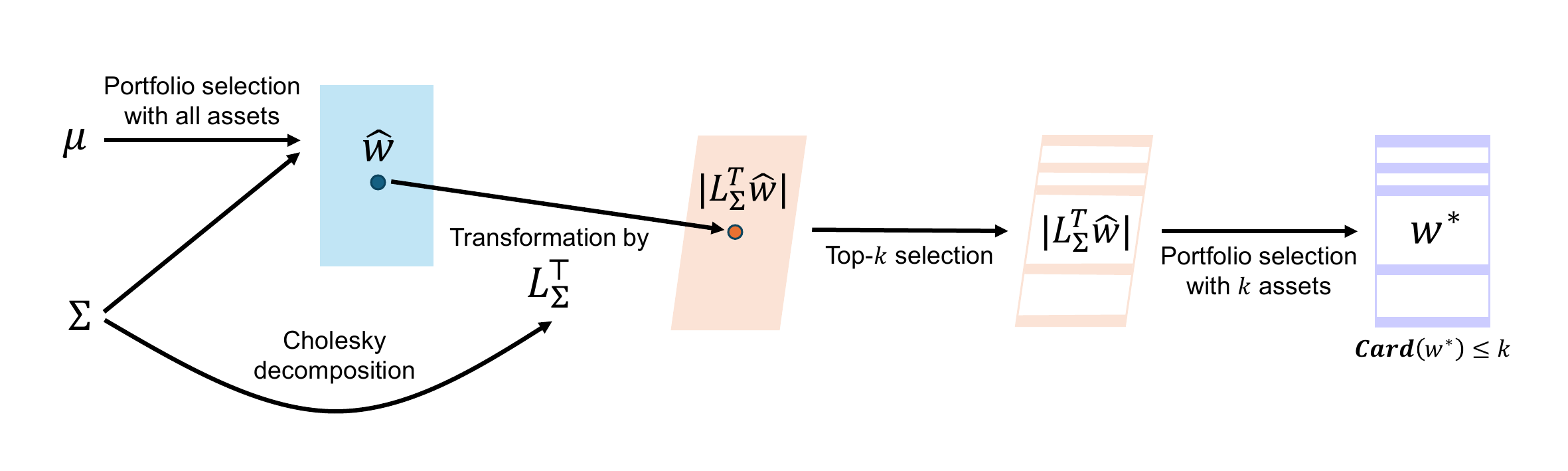}}
    \caption{
    Overview of the OSCAR heuristic for sparse tangent portfolio selection. The process starts with portfolio selection using all assets, followed by Cholesky decomposition of the covariance matrix and selection of the top-\( k \) assets based on \( |L_\Sigma^T \hat{w}| \). Finally, the constrained optimal portfolio \( w^* \) is obtained by solving the portfolio optimization problem with the selected assets.}
    \label{fig:framework}
    \end{center}
\end{figure*}
%%%%%%%%%%%%%%%%%%%%%%%%%%%%%%%%%%%%%%%%%%%%%%%%%%%%%%%%%%%%%%%%%%%%%

In this section, we first define a cardinality-constrained Sharpe ratio maximization problem.
To address the challenges posed by the cardinality constraint, we approximate the original problem.
Based on this approximation, we propose an asset selection heuristic that derives an optimal selection for an approximated problem.

\subsection{Cardinality-constrained Sharpe ratio maximization problem}
To determine the optimal tangent portfolio with the maximum Sharpe ratio \citep{sharpe1966mutual} for a set of $n$ assets, we solve the following problem, given the expected return vector and covariance matrix of returns $(\mu, \Sigma) \in (\mathbb{R}^n, \mathbb{R}^{n \times n})$.

\begin{equation}
\label{tangent}
    \begin{aligned}
        \max_{w \in \mathbb{R}^n} \quad & SR(w|\mu, \Sigma) = \frac{\mu^{T}w}{\sqrt{w^{T} \Sigma w}} \\
        \textrm{s.t.} \quad
        & \mathbf{1}^{T} w = 1 \\
    \end{aligned}
\end{equation}
where $\mathbf{1}$ denotes an $n$-dimensional vector of ones and $w$ is the weights of a portfolio. 

Although based on a solid theoretical model (\ref{tangent}), the resulting optimal portfolio is challenging to apply in real-world situations due to its inclusion of numerous well-diversified, non-zero components. A straightforward solution is to limit the number of assets in the portfolio, making it sparse. The model is then modified as follows:

\begin{equation}
\label{tangent_car}
    \begin{aligned}
        \max_{w \in \mathbb{R}^n} \quad & SR(w|\mu, \Sigma) = \frac{\mu^{T}w}{\sqrt{w^{T} \Sigma w}} \\
        \textrm{s.t.} \quad
        & \mathbf{1}^{T} w = 1 \\
        &  \mathbf{Card}(w) \leq k \\
    \end{aligned}
\end{equation}
where $k \in \{1, 2, \ldots, n\}$ is the number of assets that one can or want to hold, and $\mathbf{Card}(w)$ denotes the number of non-zero components of $w$.

Unfortunately, formulation (\ref{tangent_car}) is NP-hard \citep{garey1979computers}, meaning it cannot be solved efficiently.
Previous research addresses this issue by introducing a penalty term to the objective function \cite{brodie2009sparse} or developing efficient algorithms \cite{bienstock1996computational, li2006optimal, gao2013optimal, kim2016sparse, lee2020sparse, chang2000heuristics, maringer2003optimization}.

\subsection{OSCAR: Optimize, Select with Cholesky, And Re-optimize}
% In this section, we propose an alternative approach to selecting the assets for cardinality-constrained optimization problems.
In this section, we develop an asset selection method tailored for cardinality-constrained portfolios building on the insights from the characteristics of the Sharpe ratio presented in \cite{kim2016uniformly}.
Our method, OSCAR, utilizes the geometric properties of the Sharpe ratio and adopts Cholesky decomposition of the covariance matrix to select the number of assets beforehand to satisfy the cardinality constraint.
% We develop an asset selection method tailored for cardinality-constrained portfolios building on the insights from the characteristics of the Sharpe ratio presented in \cite{kim2016uniformly}.
% \cite{kim2016uniformly} utilized the properties of the Sharpe ratio to derive its theoretical distribution.
% Building on the insights from the characteristics of the Sharpe ratio presented in \cite{kim2016uniformly}, we develop an asset selection method tailored for cardinality-constrained portfolios.
We start by noting the Sharpe ratio is scale-invariant with respect to the portfolio weight $w$ by the following Proposition presented in \cite{kim2016uniformly}.

\begin{proposition}
\label{prop-scale-invar}
     For any $\lambda>0$, the Sharpe ratio is scale-invariant with respect to $w$, i.e.
     \[SR(w | \mu, \Sigma) = SR(\lambda w | \mu, \Sigma)\]
\end{proposition}

\begin{proof}
    See Proposition 1 in \cite{kim2016uniformly}.
\end{proof}

For any vector $w \in \mathbb{R}^n \setminus \{\mathbf{0}\}$, we can ensure that $w$ satisfies the constraint $\mathbf{1}^{T} w = \mathbf{1}$ by dividing $w$ by the sum of its components. Thus, from Proposition \ref{prop-scale-invar}, the Problem (\ref{tangent_car}) can be simplified as 

\begin{equation}
\label{tangent_car_nosumone}
    \begin{aligned}
        \max_{w \in \mathbb{R}^n} \quad & SR(w|\mu, \Sigma) = \frac{\mu^{T}w}{\sqrt{w^{T} \Sigma w}} \\
        \textrm{s.t.} \quad
        & \mathbf{Card}(w) \leq k \\
    \end{aligned}
\end{equation}

Building on this foundation, \cite{kim2016uniformly} further established criteria for assessing when the tangent portfolio outperforms alternative portfolios.

\begin{proposition}
\label{prop-udrp-angle}
    Let $w_{1}, w_{2} \in \mathbb{R}^{n}$ be any two portfolios on $n$ assets, and $\hat w \in \mathbb{R}^{n}$ be the optimal tangent portfolio from Problem (\ref{tangent}) based on the market expected excess return and the market covariance $(\mu, \Sigma) \in (\mathbb{R}^{n}, \mathbb{R}^{n \times n})$.
    Then,
    \[SR(w_{1}| \mu, \Sigma) \geq SR(w_{2}| \mu, \Sigma)\]
    if and only if
    \[\theta_{1} \leq \theta_{2}\]
    where
    \[\theta_{i} := \arccos \left(
    \frac{ (L_{\Sigma}^{T} w_{i})^{T} (L_{\Sigma}^{T} \hat w)}
    {\|L_{\Sigma}^{T} w_{i}\|_{2} \|L_{\Sigma}^{T} \hat w\|_{2}} \right)\]
    which is just an angle between $L_{\Sigma}^{T} w_{i}$ and $L_{\Sigma}^{T} \hat w$ where
    $L_{\Sigma}$ is the Cholesky decomposition of $\Sigma$
    (i.e. $\Sigma = L_{\Sigma} L_{\Sigma}^{T}$).
\end{proposition}

\begin{proof}
    See Proposition 2 in \cite{kim2016uniformly}.
\end{proof}

According to Proposition \ref{prop-udrp-angle}, we can alternatively compare portfolio Sharpe ratios by calculating the angle $\theta$.
The smaller the angle with the optimal portfolio, the higher the Sharpe ratio. 
Let $\mathcal{K}$ be a set of all subsets of $\{1, 2, \ldots, n\}$ of which the cardinality is $k$. For $K \in \mathcal{K}$, define $P_K(\mathbb{R}^n)=\{w \in \mathbb{R}^n | w_i = 0, \forall i \notin K\}$. Then, Problem (\ref{tangent_car_nosumone}) can be transformed equivalently as the following:

\[
\min_{K \in \mathcal{K}} \left( \min_{w \in P_K(\mathbb{R}^n)} \arccos{\left(
    \frac{ (L_{\Sigma}^{T} w)^{T} (L_{\Sigma}^{T} \hat w)}
    {\|L_{\Sigma}^{T} w\|_{2} \|L_{\Sigma}^{T} \hat w\|_{2}} \right)} \right)
\]

Our key idea is to approximate this problem as follows

\begin{equation}
\label{approx}
\min_{K \in \mathcal{K}} \left( \min_{w \in P_K(L_\Sigma^T(\mathbb{R}^n))} \arccos{\left(
    \frac{ (L_{\Sigma}^{T} w)^{T} (L_{\Sigma}^{T} \hat w)}
    {\|L_{\Sigma}^{T} w\|_{2} \|L_{\Sigma}^{T} \hat w\|_{2}} \right)} \right)
\end{equation}

where
\[
P_K(L_\Sigma^T(\mathbb{R}^n)) = \{w \in \mathbb{R}^n | (L_\Sigma^Tw)_i = 0, \forall i \notin K\} .
\]

Compared to Problem (\ref{tangent_car_nosumone}), Problem (\ref{approx}) can be solved easily in that the following Proposition gives the optimal solution to an approximated problem.

\begin{proposition}
\label{greedy_selection}

    Let $K^\ast \in \mathcal{K}$ be a subset corresponding to the index of the largest $k$ elements of $|L_\Sigma^T \hat w|$ where $|\cdot|$ is element-wise absolute value operator. Then $K^\ast$ is the optimal solution to Problem (\ref{approx}). i.e.
    \[
    K^\ast = \operatorname*{arg\,min}_{K \in \mathcal{K}} \left( \min_{w \in P_K(L_\Sigma^T(\mathbb{R}^n))} \arccos{\left(
    \frac{ (L_{\Sigma}^{T} w)^{T} (L_{\Sigma}^{T} \hat w)}
    {\|L_{\Sigma}^{T} w\|_{2} \|L_{\Sigma}^{T} \hat w\|_{2}} \right)} \right)
    \]    
\end{proposition}

\begin{proof}
    Let $L_\Sigma^T \hat w=(x_{1},x_{2},\dots,x_{n}) \in \mathbb{R}^n$.
    Without loss of generality, assume $|x_{1}| \geq |x_{2}| \geq \dots \geq |x_{n}|$. Then, $K^\ast = \{1, 2, \ldots, k\}$. 
    
    Since $\arccos$ is a monotonic decreasing function on $[0, 1]$, 
    \begin{align*}
        & \operatorname*{arg\,min}_{K \in \mathcal{K}} \left( \min_{w \in P_K(L_\Sigma^T(\mathbb{R}^n))} \arccos{\left(
        \frac{ (L_{\Sigma}^{T} w)^{T} (L_{\Sigma}^{T} \hat w)}
        {\|L_{\Sigma}^{T} w\|_{2} \|L_{\Sigma}^{T} \hat w\|_{2}} \right)} \right) \\
        &= 
        \operatorname*{arg\,max}_{K \in \mathcal{K}} \left( \max_{w \in P_K(L_\Sigma^T(\mathbb{R}^n))}
        \frac{ (L_{\Sigma}^{T} w)^{T} (L_{\Sigma}^{T} \hat w)}
        {\|L_{\Sigma}^{T} w\|_{2} \|L_{\Sigma}^{T} \hat w\|_{2}} \right)
    \end{align*}
        
    Let $(L_\Sigma^T \hat w)_K$ be a projection where the component of $L_\Sigma^T \hat w$ corresponding to the index not included in $K$ is changed to 0. Then, we have

    \begin{align*}
        & \operatorname*{arg\,max}_{K \in \mathcal{K}} \left( \max_{w \in P_K(L_\Sigma^T(\mathbb{R}^n))}
        \frac{ (L_{\Sigma}^{T} w)^{T} (L_{\Sigma}^{T} \hat w)}
        {\|L_{\Sigma}^{T} w\|_{2} \|L_{\Sigma}^{T} \hat w\|_{2}} \right) \\
        & = \operatorname*{arg\,max}_{K \in \mathcal{K}} \left( \max_{w \in P_K(L_\Sigma^T(\mathbb{R}^n))}
        \frac{ (L_{\Sigma}^{T} w)^{T} (L_{\Sigma}^{T} \hat w)_K}
        {\|L_{\Sigma}^{T} w\|_{2} \|L_{\Sigma}^{T} \hat w\|_{2}} \right)  \\
        &= \operatorname*{arg\,max}_{K \in \mathcal{K}} \left( \max_{w \in P_K(L_\Sigma^T(\mathbb{R}^n))}
        \frac{\|(L_{\Sigma}^{T} \hat w)_K\|_{2}}
        {\|L_{\Sigma}^{T} \hat w\|_{2}}
        \frac{ (L_{\Sigma}^{T} w)^{T} (L_{\Sigma}^{T} \hat w)_K}
        {\|L_{\Sigma}^{T} w\|_{2} \|(L_{\Sigma}^{T} \hat w)_K\|_{2}} \right) \\
        &= \operatorname*{arg\,max}_{K \in \mathcal{K}} \left(
        \frac{\|(L_{\Sigma}^{T} \hat w)_K\|_{2}}
        {\|L_{\Sigma}^{T} \hat w\|_{2}}
        \frac{ (L_{\Sigma}^{T} \hat w)_K^{T} (L_{\Sigma}^{T} \hat w)_K}
        {\|(L_{\Sigma}^{T} \hat w)_K\|_{2} \|(L_{\Sigma}^{T} \hat w)_K\|_{2}} \right) \\
        & (\because w=(L_\Sigma^T)^{-1}(L_\Sigma^T \hat w)_K \in P_K(L_\Sigma^T(\mathbb{R}^n))) \\
        &= \operatorname*{arg\,max}_{K \in \mathcal{K}} \left(
        \frac{\|(L_{\Sigma}^{T} \hat w)_K\|_{2}}
        {\|L_{\Sigma}^{T} \hat w\|_{2}} \right)
        = K^\ast
    \end{align*}

    Thus, $K^\ast = \{1, 2, \ldots, k\}$ is the optimal solution.
\end{proof}

From Proposition \ref{greedy_selection}, it can be seen that the optimal solution to Problem (\ref{approx}) is to select the assets corresponding to the $k$ largest components of $|L_\Sigma^T \hat w|$. The final step of our heuristic is to solve Problem (\ref{tangent}) with the selected assets, which yields the constrained optimal portfolio $w^\ast$. Since Problem (\ref{tangent}) has no cardinality constraint and the number of assets (dimension of decision variables) is reduced, it can be solved in a short time. We illustrate the overall procedure of OSCAR in Algorithm \ref{alg:oscar}. One important contribution in our asset selection heuristic is reusability. Let us consider a case where the asset universe ($\mu$ and $\Sigma$) remains the same, but the limit of cardinality $k$ changes to $k'$. In this case, we do not have to solve the whole problem again. We can use the same ordered asset list derived in Proposition \ref{greedy_selection}. By choosing the $k'$ number of assets instead of $k$ from the ordered asset list, we can immediately choose the best asset candidates.

Note that Problem (\ref{approx}) is an approximated problem. Therefore, the asset set obtained through Proposition \ref{greedy_selection} is not guaranteed to be optimal for Problem (\ref{tangent_car_nosumone}). Nonetheless, we demonstrate in Section \ref{sec:exp} that OSCAR achieves excellent optimality across various asset classes. Furthermore, in Section \ref{discussion}, we discuss in greater detail the relationship between OSCAR's performance and the covariance matrix of the assets.

%%%%%%%%%%%%%%%%%%%%%%% ALGORITHM %%%%%%%%%%%%%%%%%%%%%%%
\begin{algorithm}[t]
    \caption{OSCAR: Optimize, Select with Cholesky, And Re-optimize}
    \label{alg:oscar}
    \setstretch{1.1}
    \begin{algorithmic} % \small
        \State {\bfseries Input:} Expected Return $\mu$, Covariance Matrix $\Sigma$, Cardinality $k$
        \State {\bfseries Output:} Portfolio $w^*$
        % $ $(\mu, \Sigma) \in (\mathbb{R}^n, \mathbb{R}^{n \times n})$ 
        \State {\textit{1. Optimize:}}
        \State \hspace{\parindent}  Solve Problem \ref{tangent_car_nosumone} with all assets and derive $\hat{w}$.
        
        \State {\textit{2. Select with Cholesky:}}
        \State \hspace{\parindent} Decompose $\Sigma = L_{\Sigma}^{T} L_{\Sigma}$ by Cholesky decomposition.
        % \State Derive $L_{\Sigma}^{T} \hat{w}$.
        \State \hspace{\parindent} Select $k$ largest components of $| L_{\Sigma}^{T} \hat{w} |$.
        % and derive $(L_\Sigma^T \hat w)_K$.
        
        \State {\textit{3. Re-optimize:}}
        \State \hspace{\parindent} Solve Problem \ref{tangent_car_nosumone} with assets selected in Step 2 and derive $w^*$. 
    
    \end{algorithmic}
\end{algorithm}
%%%%%%%%%%%%%%%%%%%%%%%%%%%%%%%%%%%%%%%%%%%%%%%%%%%%%%%%%%%%%%%%%%%%%

\section{Numerical experiments}
\label{sec:exp}

In this section, we compare our proposed asset selection heuristic to the optimization solver CPLEX \citep{cesarone2013new, costola2022mean}.
We demonstrate how our heuristic performs compared to other widely used heuristics for cardinality-constrained problems \citep{zheng2020index}.
The evaluation focuses on three key metrics: the Sharpe ratio of the tangent portfolio, the computational time required to derive the solution, and the similarity between the assets selected by the heuristic and those chosen in the optimal solution.
% We evaluate the heuristics by measuring the Sharpe ratio of the tangent portfolio and the time taken to obtain the solution.
The ground truth is defined as the solution derived by CPLEX within one day (86,400 seconds).

\subsection{Datasets}
For our experiment, we use six 10-year daily closing price datasets from January 2014 to December 2023, obtained from Yahoo Finance.
Specifically, the datasets are EuroStoxx50 (47 assets), FTSE100 (94 assets), S\&P100 (97 assets), KOSPI200 (157 assets), Nikkei225 (213 assets) and S\&P500 (466 assets).
We exclude companies that were added to or removed from the index during the time period from our dataset.

\subsection{Evaluation Metrics and Benchmarks}
\label{sec-benchmark}
We evaluate our methodology in three aspects, \textit{performance}, \textit{time}, and \textit{hit count}:
% We define \textit{performance} as the ratio of the heuristic and CPLEX's Sharpe Ratio to quantify how well the heuristic selection methods perform relative to the optimal selection.
% \textit{Hit count} is defined as the asset matching ratio between the heuristic and CPLEX.
% These metrics provide a comprehensive assessment of heuristic effectiveness and efficiency:
\begin{itemize}
    \item[] \textbf{Performance Measures}
    \item \textit{Performance} is measured as the ratio of the Sharpe ratio achieved by the heuristic to that achieved by CPLEX, indicating the heuristic’s approximation quality relative to the optimal solution.
    \item \textit{Time} represents the computational time required to generate the solution.
    \item \textit{Hit count} is defined as the proportion of selected assets that match between the heuristic and CPLEX solutions, providing a measure of asset selection similarity.
\end{itemize}

OSCAR can be decomposed into selection and optimization steps, allowing it to be readily re-applied even when $k$ changes.
We benchmarked our heuristic against four alternative asset selection approaches commonly employed in cardinality-constrained optimization problems \cite{zheng2019indextrackingcardinalityconstraints}, which share these features.
% We benchmark our heuristic against four alternative asset selection approaches commonly employed in cardinality-constrained optimization problems \cite{zheng2019indextrackingcardinalityconstraints}:
% 오스카는 말이야, 선택과정 최적화과정 나뉘어 있고, k가 달라져도 바로 다시 적용가능해.
% 우리는 이러한 특징을 공유하면서도 cardinality opt에서 많이 사용되는 것들을 benchmark 했어~~~~ [cite]
% We chose the benchmarks where the methodology chooses $k$ assets starting from $w^*$, not solving the cardinality-constrained optimization problem directly.
% We chose the benchmarks that can be applied even when $k$ changes.
% Specifically, we selected the benchmarks that can be decomposed into the selection and the optimization steps.
% The main difficulty of the cardinality-constrained portfolio problem lies in simultaneously determining
% the asset selection that satisfies the cardinality constraint and performing the sub-sequent optimization.
\begin{itemize}
    \item[] \textbf{Baselines}
    \item Top-$k$ Sharpe Ratio Selection (SR): Solve Problem (\ref{tangent}), select the top $k$ assets with the top-$k$ Sharpe ratio value.
    Solve Problem (\ref{tangent}) with the selected assets.
    \item Top-$k$ Weight Selection (W): Solve Problem (\ref{tangent}), select the top $k$ assets with the largest absolute weight value.
    Solve Problem (\ref{tangent}) with the selected assets.
    \item Forward Selection (F): Solve Problem (\ref{tangent}), select one asset with the largest absolute weight value, record it in a selected asset list.
    Repeat the procedure without the selected assets until $k$ assets are chosen.
    \item Backward Selection (B): Solve Problem (\ref{tangent}), discard one asset with the smallest absolute weight value
    Repeat the procedure without the selected assets until $k$ assets are left.
\end{itemize}

\subsection{Results}

Here, we compare OSCAR with four other heuristics, using the ground truth derived from the CPLEX solver within 86,400 seconds.
Note that the upper bound of the performance measure is not 100\%, as the CPLEX solver may produce a suboptimal solution due to time constraints.
For the experiment, we set four cardinality constraints $k$ as \{5, 10, 15, 20\}\% of each dataset, rounding up to the nearest integer.

Table \ref{table:k5-20} summarizes the computation time and performance of each dataset, using both the heuristics and the CPLEX solver.
The number of assets $N$ is specified for each dataset, and each table represents a different cardinality constraint $k$.
% For each dataset, we denoted the number of assets $N$.
% Each table indicates different cardinality constraint $k$.
The time and performance are indicated for four benchmarks and OSCAR for each dataset.
% We denote time for the ground truth using CPLEX with the time constraint of 86400 seconds.

Notably, the CPLEX was unable to derive an optimal solution within a day solving for the S\&P100 dataset (97 assets) with $k=0.10 \times N$ as shown in Table \ref{table:k5-20}.
Furthermore, starting from $k=0.15 \times N$, CPLEX failed to terminate across all datasets except for the EuroStoxx50.
% As the values of $N$ and $k$ increased, the solver produced suboptimal solutions.

To provide a more intuitive understanding of the performance-time trade-off, we present a scatter plot in  Figure \ref{fig:scatter}.
Each heuristic is represented with a distinct marker, with OSCAR highlighted as a red dot.
The horizontal dashed line at a performance level of 100\% represents the CPLEX result.
The time is log-scaled.
The upper-left region of the figure indicates higher performance with shorter computation time.

OSCAR, as well as the top-$k$ Sharpe ratio (SR) and top-$k$ weight (W) selection methods, required less computation time compared to the forward (F) and backward (B) selection methods.
This difference is attributed to selecting $k$ assets in a single loop, whereas forward (F) and backward (B) selection methods iteratively solve optimization problems.

%%%%%%%%%% TABLE: Results %%%%%%%%%%
% \resizebox{\columnwidth}{!}{
\begin{sidewaystable}[htp]
    \centering
    % First Table

    \caption{
    Performance and time results for benchmarks, OSCAR, and CPLEX under cardinality constraints \(k = \ceil{ 0.05 \times N }\), \(k = \ceil{ 0.10 \times N }\), \(k = \ceil{ 0.15 \times N }\), and \(k = \ceil{ 0.20 \times N }\).
    }
    \label{table:k5-20}
    
    % \scalebox{0.9}{
    % \resizebox{\columnwidth}{!}{
        \setlength{\tabcolsep}{4pt}
        \begin{tabular}{cccccccccccccc}
        
        \toprule
        \multirow{2}{*}{\textbf{Data}} & \multirow{2}{*}{\textbf{N}} 
        & \multicolumn{6}{c}{\textbf{\(k = \ceil{ 0.05 \times N }\)}} 
        & \multicolumn{6}{c}{\textbf{\(k = \ceil{ 0.10 \times N }\)}} \\
        \cmidrule(lr){3-8} \cmidrule(lr){9-14}
               &  & SR & W & F & B & OSCAR & CPLEX
               & SR & W & F & B & OSCAR & CPLEX \\
        \midrule
        EuroStoxx50 & 47
         & \makecell{$0.08s$\\83.17\%} & \makecell{$0.09s$\\87.65\%} & \makecell{$0.07s$\\62.98\%} & \makecell{$0.49s$\\98.58\%} & \makecell{$0.03s$\\86.30\%} & $0.87s$ 
         & \makecell{$0.07s$\\76.07\%} & \makecell{$0.07s$\\94.59\%} & \makecell{$0.09s$\\79.35\%} & \makecell{$0.54s$\\99.22\%} & \makecell{$0.03s$\\93.15\%} & $2.48s$ \\
        \hline
        
        FTSE100 & 94
         & \makecell{$0.07s$\\75.25\%} & \makecell{$0.10s$\\82.20\%} & \makecell{$0.22s$\\56.28\%} & \makecell{$1.53s$\\91.01\%} & \makecell{$0.06s$\\99.05\%} & $449.86s$ 
         & \makecell{$0.07s$\\67.09\%} & \makecell{$0.10s$\\95.00\%} & \makecell{$0.30s$\\59.41\%} & \makecell{$1.40s$\\97.66\%} & \makecell{$0.17s$\\93.94\%} & $51254.78s$ \\
        \hline
        
        S\&P100 & 97
         & \makecell{$0.08s$\\86.15\%} & \makecell{$0.14s$\\83.26\%} & \makecell{$0.33s$\\66.73\%} & \makecell{$1.87s$\\99.87\%} & \makecell{$0.08s$\\95.13\%} & $358.34s$ 
         & \makecell{$0.07s$\\78.01\%} & \makecell{$0.10s$\\91.87\%} & \makecell{$0.62s$\\70.83\%} & \makecell{$1.89s$\\99.06\%} & \makecell{$0.09s$\\95.88\%} & $86400s$ \\
        \hline
        
        KOSPI200 & 157
        & \makecell{$0.12s$\\81.49\%} & \makecell{$0.10s$\\74.96\%} & \makecell{$0.94s$\\58.36\%} & \makecell{$8.48s$\\93.14\%} & \makecell{$0.11s$\\95.56\%} & $86400s$ 
        & \makecell{$0.07s$\\72.20\%} & \makecell{$0.12s$\\88.00\%} & \makecell{$1.90s$\\53.42\%} & \makecell{$10.27s$\\94.07\%} & \makecell{$0.10s$\\97.03\%} & $86400s$ \\
        \hline
        
        NIKKEI225 & 213
        & \makecell{$0.21s$\\63.22\%} & \makecell{$0.18s$\\71.71\%} & \makecell{$1.67s$\\64.70\%} & \makecell{$17.67s$\\88.06\%} & \makecell{$0.12s$\\89.63\%} & $86400s$ 
        & \makecell{$0.07s$\\54.80\%} & \makecell{$0.15s$\\94.16\%} & \makecell{$3.17s$\\62.85\%} & \makecell{$16.82s$\\94.26\%} & \makecell{$0.18s$\\94.09\%} & $86400s$ \\
        \hline
        
        S\&P500 & 466
        & \makecell{$0.45s$\\67.15\%} & \makecell{$0.40s$\\75.25\%} & \makecell{$8.84s$\\70.06\%} & \makecell{$78.50s$\\99.99\%} & \makecell{$0.51s$\\95.32\%} & $86400s$ 
        & \makecell{$0.46s$\\63.03\%} & \makecell{$0.58s$\\84.41\%} & \makecell{$16.04s$\\72.30\%} & \makecell{$79.62s$\\95.15\%} & \makecell{$1.80s$\\95.51\%} & $86400s$ \\
        \hline & \\
        
        %%%%%%%%%%%% Second table %%%%%%%%%%%%
        % \resizebox{\columnwidth}{!}{
        \toprule

        \multirow{2}{*}{\textbf{Data}} & \multirow{2}{*}{\textbf{N}} 
        & \multicolumn{6}{c}{\textbf{\(k = \ceil{ 0.15 \times N }\)}} 
        & \multicolumn{6}{c}{\textbf{\(k = \ceil{ 0.20 \times N }\)}} \\
        \cmidrule(lr){3-8} \cmidrule(lr){9-14}
               &  & SR & W & F & B & OSCAR & CPLEX 
               & SR & W & F & B & OSCAR & CPLEX \\
        \midrule
        EuroStoxx50 & 47
        & \makecell{$0.08s$\\70.77\%} & \makecell{$0.08s$\\97.54\%} & \makecell{$0.18s$\\79.28\%} & \makecell{$0.44s$\\94.90\%} & \makecell{$0.04s$\\94.47\%} & $7.94s$ 
        & \makecell{$0.08s$\\69.33\%} & \makecell{$0.07s$\\97.46\%} & \makecell{$0.14s$\\78.29\%} & \makecell{$0.46s$\\97.49\%} & \makecell{$0.04s$\\93.75\%} & $14.85s$ \\
        \hline
        FTSE100 & 94
        & \makecell{$0.09s$\\64.26\%} & \makecell{$0.09s$\\96.28\%} & \makecell{$0.45s$\\59.09\%} & \makecell{$1.45s$\\98.76\%} & \makecell{$0.09s$\\96.89\%} & $86400s$ 
        & \makecell{$0.09s$\\63.43\%} & \makecell{$0.09s$\\94.99\%} & \makecell{$0.55s$\\60.04\%} & \makecell{$1.20s$\\100.24\%} & \makecell{$0.09s$\\94.44\%} & $86400s$ \\
        \hline
        S\&P100 & 97
        & \makecell{$0.11s$\\73.19\%} & \makecell{$0.08s$\\91.35\%} & \makecell{$0.66s$\\66.73\%} & \makecell{$1.94s$\\99.87\%} & \makecell{$0.09s$\\95.89\%} & $86400s$ 
        & \makecell{$0.16s$\\70.87\%} & \makecell{$0.09s$\\94.99\%} & \makecell{$0.73s$\\64.70\%} & \makecell{$1.78s$\\100.24\%} & \makecell{$0.17s$\\95.87\%} & $86400s$ \\
        \hline
        KOSPI200 & 157
        & \makecell{$0.07s$\\67.59\%} & \makecell{$0.14s$\\91.90\%} & \makecell{$2.87s$\\59.01\%} & \makecell{$8.81s$\\94.61\%} & \makecell{$0.09s$\\97.07\%} & $86400s$ 
        & \makecell{$0.12s$\\65.77\%} & \makecell{$0.11s$\\90.75\%} & \makecell{$3.71s$\\58.07\%} & \makecell{$7.94s$\\92.14\%} & \makecell{$0.17s$\\97.07\%} & $86400s$ \\
        \hline
        NIKKEI225 & 213
        & \makecell{$0.14s$\\52.79\%} & \makecell{$0.18s$\\91.35\%} & \makecell{$4.39s$\\67.29\%} & \makecell{$17.41s$\\98.74\%} & \makecell{$0.14s$\\94.47\%} & $86400s$ 
        & \makecell{$0.12s$\\51.97\%} & \makecell{$0.22s$\\90.75\%} & \makecell{$7.46s$\\62.31\%} & \makecell{$15.81s$\\100.05\%} & \makecell{$0.16s$\\94.47\%} & $86400s$ \\
        \hline
        S\&P500 & 466
        & \makecell{$0.44s$\\61.91\%} & \makecell{$0.42s$\\96.03\%} & \makecell{$24.77s$\\70.73\%} & \makecell{$78.17s$\\100.50\%} & \makecell{$0.46s$\\92.13\%} & $86400s$ 
        & \makecell{$0.44s$\\61.55\%} & \makecell{$0.47s$\\94.99\%} & \makecell{$30.46s$\\67.41\%} & \makecell{$77.34s$\\100.69\%} & \makecell{$0.50s$\\91.27\%} & $86400s$ \\
        \hline
        
        \end{tabular}
    \setlength{\tabcolsep}{6pt}
    % }
\end{sidewaystable}
\begin{figure}[t]
\includegraphics[width=\textwidth]{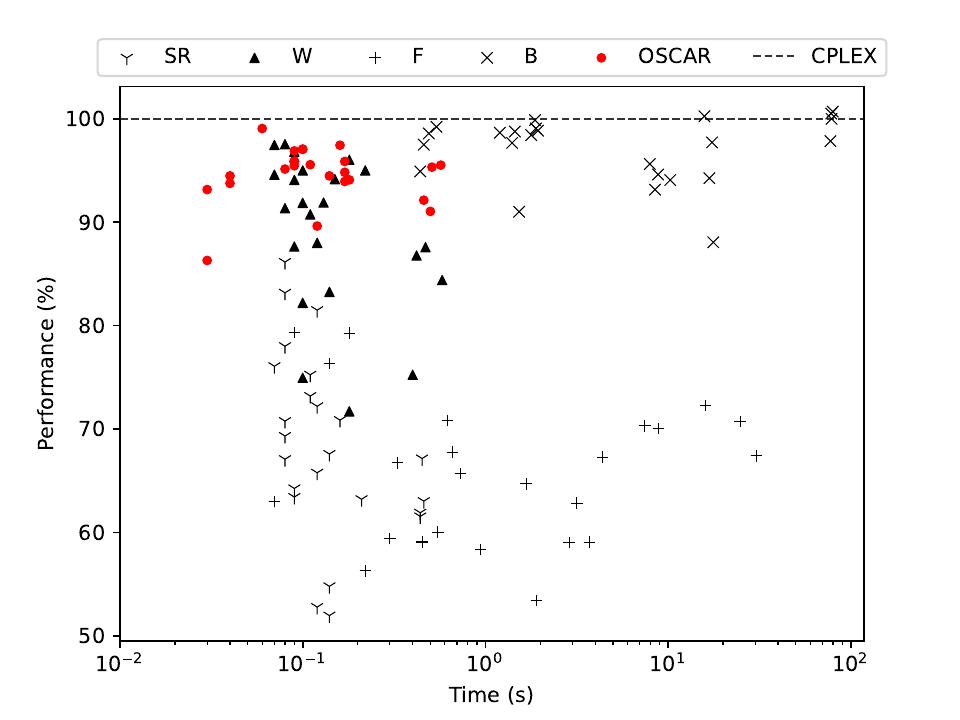}
\centering
\caption{
Scatter plot illustrating performance (\%) and time (s) of benchmarks and OSCAR.
The time axis is log-scaled.
Points on the upper-left region represent higher performance and shorter asset selection times.
The horizontal dashed line represents the CPLEX result with the time limit of 86400 seconds.
Note that the CPLEX result may not represent the true optimal solution.
}
\label{fig:scatter}
\end{figure}
%%%%%%%%%%%%%%%%%%%%%%%%%%%%%%%%%%%%

From a performance standpoint, OSCAR, W and B outperformed SR and F.
In some instances, heuristic B even surpassed the CPLEX solver's results, achieving a performance exceeding 100\%.
However, the backward selection method was the slowest among the benchmarks, taking nearly $10^2$ times longer than OSCAR, SR, and W.
Furthermore, the computation time for F and B selection methods increased proportionally as the number of assets grow, due to the increasing number of solving optimization problems.

In contrast, OSCAR maintained both efficiency and efficacy even with the increasing number of assets.
While the time required for Cholesky decomposition does increase with the number of assets $N$, this overhead is negligible compared to the time consumed by solving optimization problems.
Importantly, OSCAR demonstrated consistent computation times across varying cardinality constraints $k$.
From a performance-time trade-off point of view, we conclude that the portfolio derived by OSCAR satisfies both high performance and short time compared to other heuristics.

% We denote the case where CPLEX could not solve the problem within a day as the "-" mark.
% It can be seen that the computation time increases exponentially as the number of assets $N$ and cardinality $k$ increase. Also, we marked the cases of achieving 100\% performance as "\checkmark".
% As shown in the tables, OSCAR was the fastest and produced the highest quality portfolio among the heuristics in most cases. Also, the performance of OSCAR was near ground truth for all cases. In particular, for the problem of choosing $8$ assets ($k = 8$) in the S\&P500 dataset, OSCAR derived a portfolio that has about 96\% performance compared to the ground truth only in 1.33 seconds. This is much shorter than the computation time to obtain the ground truth in CPLEX, which was 86400 seconds. There were some cases when Backward derived a higher quality portfolio. However, the time required for the Backward method increases exponentially as the number of assets grows, due to its inherent characteristics. Also, it takes significantly more time compared to OSCAR. Even as the number of assets increases, OSCAR maintains both efficiency and efficacy, demonstrating consistent performance with nearly the same computation time. From a performance-time trade-off point of view, we can conclude that the portfolio derived by OSCAR satisfies both high quality and short time compared to other heuristics.

We also present the results from the hit count perspective in Table \ref{table:hitcount}.
For six datasets, we evaluated four cardinality-constrained cases for each dataset.
For the chosen $k$ number of assets by CPLEX from total $N$ assets, we present the number of matching assets achieved by heuristics.
We bold-lettered and underlined for the best performing results for each experiment.

Notably, OSCAR achieved a high hit count in the majority of cases, outperforming other benchmarks in most instances.
Even in cases where OSCAR did not achieve the highest hit count, its performance was consistently close to the best.

\begin{table*}[t]
    \caption{
    Hit count for top-$k$ Sharpe ratio (SR), top-$k$ weight (W), forward (F), backward (B) selections and OSCAR.
    We use four cardinality constraints $k$ as \{5, 10, 15, 20\}\% of each dataset's asset number $N$, rounding up to the nearest integer.
    The best performing results in each experiment are bold-lettered and underlined.
    }
    \centering
    \begin{small}
        \begin{tabularx}{\textwidth}{cc @{\hspace*{10mm}} c @{\hspace*{10mm}} *5{Y}}
        \toprule
        \multirow{2}{*}{\textbf{Data}} & \multirow{2}{*}{\textbf{N}} & \multirow{2}{*}{\textbf{$k$}} & \multicolumn{5}{c}{\textbf{Benchmarks}} \\
        \cmidrule(l){4-8}
               & &  & SR  & W & F & B & OSCAR \\
        \midrule
        EuroStoxx50 & 47 & \makecell{3\\5\\8\\10}
        & \makecell{\underline{\textbf{2}}\\2\\3\\4}
        & \makecell{\underline{\textbf{2}}\\\underline{\textbf{4}}\\\underline{\textbf{6}}\\\underline{\textbf{7}}}
        & \makecell{0\\2\\2\\3}
        & \makecell{\underline{\textbf{2}}\\\underline{\textbf{4}}\\5\\\underline{\textbf{7}}}
        & \makecell{\underline{\textbf{2}}\\\underline{\textbf{4}}\\4\\6} \\
        \hline
        FTSE100 & 94 & \makecell{5\\10\\15\\19}
        & \makecell{2\\4\\6\\8}
        & \makecell{2\\\underline{\textbf{8}}\\12\\13}
        & \makecell{1\\2\\3\\5}
        & \makecell{3\\\underline{\textbf{8}}\\\underline{\textbf{14}}\\\underline{\textbf{15}}}
        & \makecell{\underline{\textbf{4}}\\\underline{\textbf{8}}\\11\\\underline{\textbf{15}}} \\
        \hline
        S\&P100 & 97 & \makecell{5\\10\\15\\20}
        & \makecell{\underline{\textbf{3}}\\4\\5\\6}
        & \makecell{2\\7\\8\\15}
        & \makecell{0\\4\\4\\6}
        & \makecell{\underline{\textbf{3}}\\\underline{\textbf{8}}\\\underline{\textbf{10}}\\\underline{\textbf{16}}}
        & \makecell{2\\\underline{\textbf{8}}\\\underline{\textbf{10}}\\\underline{\textbf{16}}} \\
        \hline
        KOSPI200 & 157 & \makecell{8\\16\\24\\32}
        & \makecell{4\\7\\10\\15}
        & \makecell{3\\8\\16\\22}
        & \makecell{1\\2\\5\\9}
        & \makecell{\underline{\textbf{5}}\\10\\14\\20}
        & \makecell{4\\\underline{\textbf{12}}\\\underline{\textbf{19}}\\\underline{\textbf{27}}} \\
        \hline
        NIKKEI225 & 213 & \makecell{11\\22\\32\\43}
        & \makecell{4\\8\\14\\16}
        & \makecell{2\\\underline{\textbf{15}}\\\underline{\textbf{25}}\\34} 
        & \makecell{3\\6\\8\\13} 
        & \makecell{\underline{\textbf{5}}\\13\\22\\\underline{\textbf{35}}}
        & \makecell{\underline{\textbf{5}}\\\underline{\textbf{15}}\\24\\34} \\
        \hline

        S\&P500 & 466 & \makecell{24\\47\\70\\94}
        & \makecell{4\\11\\17\\25}
        & \makecell{7\\15\\27\\55}
        & \makecell{4\\5\\5\\15} 
        & \makecell{\underline{\textbf{10}}\\16\\32\\\underline{\textbf{59}}}
        & \makecell{\underline{\textbf{10}}\\\underline{\textbf{24}}\\\underline{\textbf{37}}\\57} \\
        \hline
        \end{tabularx}
    \end{small}
    \label{table:hitcount}
\end{table*}

\begin{figure*}[t!]
\begin{center}
\centerline{\includegraphics[width=0.8\textwidth]{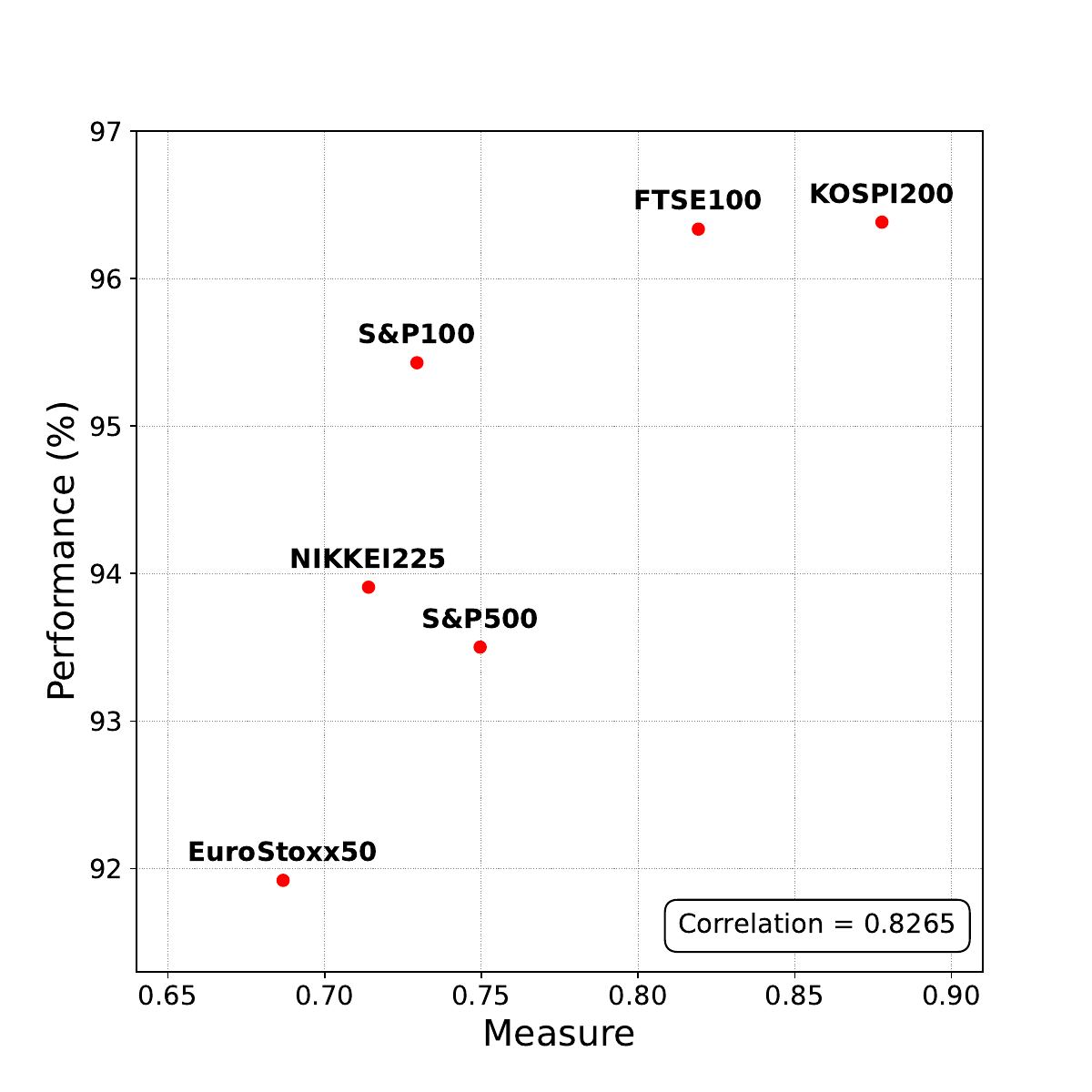}}
\caption{Scatter plot illustrating the relationship between OSCAR’s performance and the diagonal dominance of asset set covariance matrices. The x-axis represents the ratio of the mean diagonal elements to the sum of the mean diagonal and mean off-diagonal elements, indicating the degree of diagonal dominance in each covariance matrix. The y-axis shows OSCAR’s average performance across various cardinality constraints, measured as the Sharpe ratio relative to the optimal selection. The plot demonstrates that higher diagonal dominance, which reflects lower correlations between assets, generally leads to better OSCAR performance.}
\label{fig:measure}
\end{center}
\end{figure*}

\section{Discussion}
\label{discussion}

In section \ref{sec:exp}, we demonstrated the performance of OSCAR with three different evaluation metrics: performance, time, and hit count.
In this section, we further explore the relationship between the characteristics of asset sets and OSCAR’s performance, providing insights into the underlying factors that influence its effectiveness.

In Problem (4), we approximated the original space \( P_K(\mathbb{R}^n) \) by transforming it into \( P_K(L_{\Sigma}^{T}\mathbb{R}^n) \).
Specifically, this transformation yields:
\[
P_K(L_\Sigma^T(\mathbb{R}^n)) = \{w \in \mathbb{R}^n \mid (L_\Sigma^Tw)_i = 0, \;\forall i \notin K\},
\]
implying that as \( L_{\Sigma}^{T} \) approaches a diagonal matrix, the resulting weights \( w \) increasingly resemble the optimal tangent portfolio \( \hat w \).
In the special case where \( L_{\Sigma}^{T} \) becomes a perfect diagonal matrix, the solution exactly mirrors \( \hat w \). 
This behavior of \( \Sigma \) affects the structure of the weight vector \( w \), which in turn directly impacts OSCAR’s performance.
Furthermore, as \( L_{\Sigma}^{T} \) more closely resembles a diagonal matrix, i.e., when the absolute values of \( L_{\Sigma}^{T} \)'s off-diagonal elements become significantly smaller than those of the diagonal elements, the corresponding covariance matrix \( \Sigma = L_{\Sigma} L_{\Sigma}^{T} \) will also exhibit off-diagonal elements with similarly small absolute values.

To examine this in more detail, we analyze how OSCAR’s performance correlates with the average ratio of diagonal to off-diagonal elements in the covariance matrices of different asset sets.
We devise a measure for this ratio as:
\[
\frac{\text{mean of diagonal elements of } |\Sigma|}{\text{mean of diagonal elements of } |\Sigma| + \text{mean of off-diagonal elements of } |\Sigma|}.
\]
We hypothesize that asset sets with higher values of the defined measure will exhibit improved performance with OSCAR. This is confirmed by Figure \ref{fig:measure}, which shows a strong positive correlation (0.8265) between the diagonal dominance of the covariance matrices and OSCAR’s performance. Specifically, the figure demonstrates that as the diagonal dominance measure increases, indicating lower interdependence among assets, OSCAR consistently achieves higher Sharpe ratios. For example, asset sets like FTSE100 and KOSPI200, which show higher diagonal dominance, consistently outperform others, highlighting the heuristic's effectiveness in less correlated environments.

The results suggest that OSCAR performs better when asset interdependence is lower, as indicated by the smaller off-diagonal elements in the covariance matrix. This relationship highlights that reduced interdependence among assets enhances the effectiveness of OSCAR’s asset selection heuristic.

\section{Conclusion}
\label{sec:conclusion}

In this paper, we propose a novel asset selection heuristic named OSCAR for sparse tangent portfolio selection. Given the expected return vector and covariance matrix, OSCAR solves the cardinality-constrained tangent portfolio selection problem by first optimizing the unconstrained version and deriving the optimal tangent portfolio \( \hat w \). Then, OSCAR selects asset candidates corresponding to the \( k \) largest components of \( |L_\Sigma \hat w| \), where \( L_\Sigma \) is the Cholesky decomposition of the covariance matrix.
Finally, OSCAR re-optimizes the portfolio optimization problem without cardinality constraint using the selected assets, yielding the constrained optimal portfolio $w^\ast$.
Our numerical experiments demonstrate that OSCAR performs with higher efficacy and efficiency compared to other heuristics.
We further validate our assumption regarding the conditions under which the OSCAR performs well by comparing the diagonal dominance of the covariance matrix and OSCAR's performance.
% Since OSCAR provides an approximate solution, future research will focus on identifying the conditions under which the approximation is most accurate.
For future research, we aim to extend OSCAR to solve problems with convex constraints, such as no-short-selling constraints.
We anticipate that OSCAR can be further developed and become a fundamental methodology for solving sparse tangent portfolio selection problems.

% Here, by adding the hyperparameter $l$, we can obtain both the quality and the robustness of the portfolio. 

% \section*{Acknowledgement}
% This work was supported by the National Research Foundation of Korea (NRF) grant funded by the Korean government (MSIT) (No. NRF-2022M3J6A1063021, No. RS-2023-00208980 and No. RS-2023-00249931).

% \section*{Declarations}
% The authors declare that there are no competing interests.

% \section*{Data availability}
% Data will be made available on request.

% \begin{appendices}

% \section{Section title of first appendix}\label{secA1}

% An appendix contains supplementary information that is not an essential part of the text itself but which may be helpful in providing a more comprehensive understanding of the research problem or it is information that is too cumbersome to be included in the body of the paper.

% %%=============================================%%
% %% For submissions to Nature Portfolio Journals %%
% %% please use the heading ``Extended Data''.   %%
% %%=============================================%%

% %%=============================================================%%
% %% Sample for another appendix section			       %%
% %%=============================================================%%

% %% \section{Example of another appendix section}\label{secA2}%
% %% Appendices may be used for helpful, supporting or essential material that would otherwise 
% %% clutter, break up or be distracting to the text. Appendices can consist of sections, figures, 
% %% tables and equations etc.

% \end{appendices}

\clearpage
\bibliography{ref}
\end{document}